\documentclass[draftclsnofoot,12pt,onecolumn] {IEEEtran}

\usepackage{mdwlist}

\usepackage{flushend}
\usepackage{url}
\usepackage{graphicx}
\usepackage{subfigure}
\usepackage{xcolor,colortbl}
\usepackage[mathscr]{eucal}
\usepackage{amsthm }
\usepackage{amssymb}

\usepackage{cases}

\usepackage{mathptmx}

\usepackage{wrapfig}

\usepackage{graphicx}
\usepackage{subfigure}
\usepackage{algorithm}
\usepackage[noend]{algorithmic}

\newtheorem{example}{Example}[section]

\newtheorem{corollary}{Corollary}[section]
\newtheorem{definition}{Definition}[section]
\newtheorem{lemma}{Lemma}[section]

\newtheorem{theorem}{Theorem}[section]
\floatname{algorithm}{Algorithm}

\usepackage{mathptmx}
\usepackage{mdwlist}
\usepackage{flushend}
\usepackage[mathscr]{eucal}
\usepackage{xcolor,colortbl}
\usepackage{url}
\ifCLASSOPTIONcompsoc
\else
\fi
\ifCLASSINFOpdf
\else
\fi
\begin{document}

%
\title{\huge Finding Efficient Region in The Plane with Line segments  }
%
%
%
%

\author{Jack~Wang 
\IEEEcompsocitemizethanks{\IEEEcompsocthanksitem Contact information: cszjwang@gmail.com. 
\protect\\

}
\thanks{}}

\IEEEcompsoctitleabstractindextext{%
\begin{abstract}
Let $\mathscr O$ be a set of $n$ disjoint  obstacles in $\mathbb{R}^2$, $\mathscr M$ be a moving object. Let $s$ and $l$ denote  the starting point  and   maximum path length  of the moving object $\mathscr M$, respectively.  Given a point $p$ in ${R}^2$, we say the point $p$ is achievable for $\mathscr M$  such that $\pi(s,p)\leq l$, where $\pi(\cdot)$ denotes the shortest path length  in the presence of obstacles.  One is to  find   a region  $\mathscr R$ such that,   for any point  $p\in \mathbb{R}^2$, if it is achievable for $\mathscr M$, then $p\in \mathscr R$; otherwise,  $p\notin \mathscr R$.   In this paper, we restrict our attention to the case of line-segment obstacles.  To tackle this problem, we develop three algorithms. We first present  a simpler-version algorithm for the sake of intuition. Its basic idea is to reduce our problem to computing the union of a set of circular visibility regions (CVRs).  This algorithm takes $O(n^3)$ time. By analysing its dominant steps, we break through its bottleneck   by using the short path map (SPM) technique to  obtain those  circles (unavailable beforehand),  yielding an $O(n^2\log n)$ algorithm. Owing to  the finding above, the third algorithm also uses the SPM technique. It however, does not continue to construct the CVRs. Instead, it directly traverses each region of the SPM to trace the boundaries, the final algorithm obtains  $O(n\log n)$ complexity.
\end{abstract}

}

\maketitle

\IEEEdisplaynotcompsoctitleabstractindextext

%
\IEEEpeerreviewmaketitle





\section{Introduction}\label{sec:1}





Suppose there are a set ${\mathscr O}$ of $n$ disjoint obstacles and a moving object $\mathscr M$  in  $\mathbb{R}^2$, and suppose   $\mathscr M$  freely moves  in $\mathbb{R}^2$ except that it cannot be allowed to directly pass through any obstacle $o\in \mathscr O$. 
See the right figure for example. The black line-segments denote the set of obstacles, and the black dot  denotes   the starting point of $\mathscr M$,  the grey line-segment denotes the maximum path length  that $\mathscr M$ is allowed to travel. We address the  problem,  how to find  a region $\mathscr R$ such that, for any point $p$  if $\mathscr M$ can reach  it, then $p\in \mathscr R$; otherwise, $p\notin \mathscr R$.  (See Section 2 for  more formal definitions and other  constraint conditions.)

\begin{figure}[h]
  \centering
\includegraphics[scale=.35]{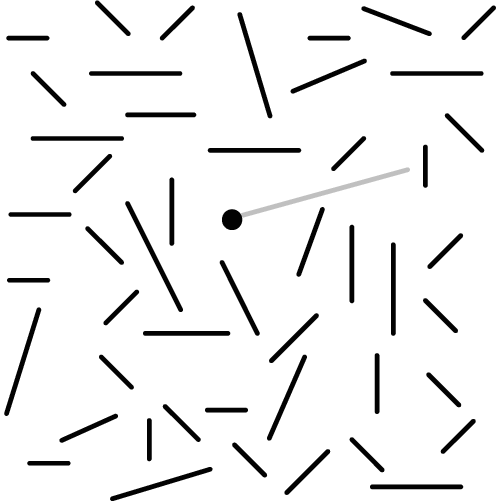}
 \caption{\small Example of Finding Achievable Region} 
 \label{fig:far}
\end{figure}

Clearly, if there is no obstacle in $\mathbb{R}^2$, the answer  is a circle, denoted by $C(s,l)$, where $s$ and $l$ denote the center and radius of the circle, respectively. Now consider the case of \textit{one} obstacle as shown in Figure \ref{fig:1:a}. Firstly, we can easily know any point in  the  region bounded by the solid lines   is achievable  for  $\mathscr M$,    without the need of making a turn, see Figure \ref{fig:1:b}. 
Secondly, we  can also easily know any point in the circle ${C}(a,dist(a,v_1))$ is  achievable for $\mathscr M$, where $dist(\cdot)$  denotes the Euclidean distance, see Figure \ref{fig:1:c}.  Similarly, we can get another circle  ${C}(b,dist(b,v_2))$, see Figure \ref{fig:1:d}. Naturally, we can get the answer by merging the three regions, i.e., two circles and a circular-arc polygon.

By investigating the simplest case, we seemingly can derive a rough solution called \textsf{RS} as follows. Firstly, we   obtain the region denoted by ${\mathscr R}_d$ in which any point is achievable for $\mathscr M$, without the need of making a turn. Second, for each vertex (or endpoint) $v$ of obstacles, if $\pi (s,v)<l$, we obtain the circle centered at  $v$ and with the radius $l- \pi (s,v)$, where $\pi (\cdot)$ denotes the shortest path length in the presence of obstacles.  Finally, we merge all the regions obtained in the previous two steps.  Is it  really so simple? (See Section \ref{sec:problem definition} for a more detailed analysis.)

\begin{figure}[h]
  \centering
  \subfigure[\scriptsize {  } ]{\label{fig:1:a}
     \includegraphics[scale=.29]{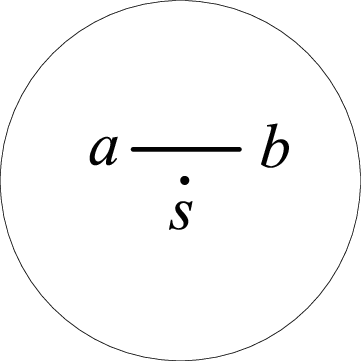}} 
  \subfigure[\scriptsize { }]{\label{fig:1:b}
     \includegraphics[scale=.29]{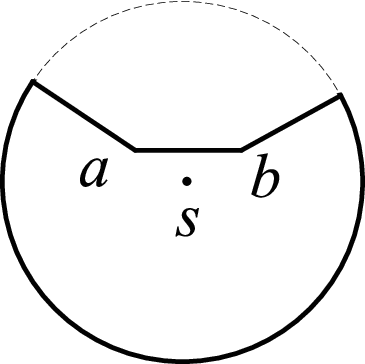}} 
  \subfigure[\scriptsize { }]{\label{fig:1:c}
    \includegraphics[scale=.29]{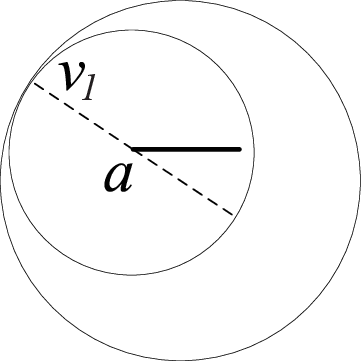}} 
   \subfigure[\scriptsize { }]{\label{fig:1:d}
        \includegraphics[scale=.29]{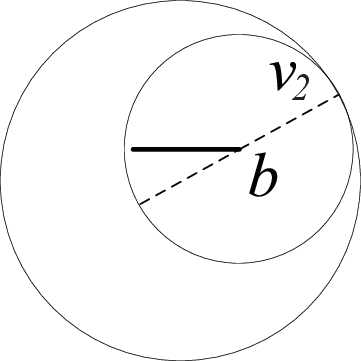}} 
 \caption{\small The simplest case. The black   dot $s$ and solid line $\overline{ab}$  denote  the starting point of $\mathscr M$ and the obstacle, respectively. } 
 \label{fig:1a}
\end{figure}




\paragraph*{\textbf{\upshape Motivations}}

Our study is motivated by the increasing popularity of \textit{location-based service} \cite{EvaggeliaPitoura:Locating,JochenSchiller:location} and \textit{probabilistic answer  on uncertain data} \cite{jianpei:queryAnswering,charuCAggarwal:aSurvey}. Nowadays, many mobile devices, such as cell phones, PDAs, taxies are equipped with GPS. It is common for the  service  provider to keep track of the user' s  location, and provide many location-based services, for instance  finding the taxies being located in a query range, finding the nearest supermarkets or ATMs, etc. The database sever, however, is often impossible  to contain the total status of an entity, due to the limited network bandwidth and limited battery power of the mobile devices. This implies that the current  location of an entity is  uncertain before obtaining the next specific location. In this case, it is meaningful to use a closed region to represent the possible location of an entity  \cite{reynoldcheng:queryingImprecise,aprasadsistal:Modeling}. In database community, they usually assumed  this region   is available beforehand, which is impractical. A proper manner to obtain such a region is generally based on the following information: geographical information around the entity, say ${G_I}$,  the (maximum) speed of the entity, say $V_M$, and the elapsed time, say $T_E$. By   substituting  ``$V_M\cdot T_E$''  with the maximum path length $l$, ${G_I}$ with obstacles $\mathscr O$, it  corresponds to our problem.

Our problem also finds applications in the so-called \textit{moving target search} \cite{andreaskolling:computing,xiaoxunsun:efficient}. Traditionally, moving target search is the problem where a hunter has to catch a moving target, and they  assumed the hunter always knows the current position of the moving target \cite{xiaoxunsun:efficient}. In many scenarios (e.g., when the power of GPS - equipped with the moving target- is used up, or in a sensor network environment, when the moving target walked out of the scope of being  minitored), it is possible that the current position of the moving target cannot be obtained. In this case, we can infer the available region based on some available knowledge such as the geographical information and the previous location. The available region here can contribute to  the reduction of search range. 

\paragraph*{\textbf{\upshape Related work}}
Although our problem is easily stated and can  find many applications,  to date, we are not aware of any published result. Our problem is generally falls in the realm of computational geometry. 
In this community, the problem of computing the visibility polygon (VP, a.k.a., visibility region)  \cite{PaulJHeffernan:anoptimal,TakaoAsano:anefficient,subhansh:worstcase}    is the  most similar to our problem. Given a point $p$ and a set of  obstacles  in a plane, this problem is to find a region in which each point is visible from $p$.    There are three major differences between our problem and the VP  problem: (\romannumeral 1) our problem has an extra constraint, i.e., the maximum path length $l$, whereas the VP problem does not involve this concept; (\romannumeral 2) in the VP problem, the region behind the obstacle must not be visible, whereas in our problem it  may be achievable by making a turn; and (\romannumeral 3) the VP  problem does not involve  the circular arc segments, whereas in our problem we have to handle a number of circular arc segments.    It is easy to see that our problem is more complicated. 


We note that the Euclidean shortest path problem \cite{DTLee:euclidean,SanjivKapoor:anEfficient,TakaoAsano:Visibility,MichaSharir:onShortest,JosephSBMitchell:shortest,SanjivKapoor:efficient,HansRohnert:shortest,JamesAStorer:shortest} also involves the obstacles and  the path length. In  these two aspects it is similar to our problem.  Given two points $p$ and $p^*$, and a set of  obstacles in a plane, this problem is to find the shortest path between $p$ and $p^*$ that does not intersect with any of the obstacles.  It is different from our problem in three points at least: (\romannumeral 1) this problem is to find a path, whereas our problem is to find a region; (\romannumeral 2) the starting point and ending point are known beforehand in this problem, whereas the ending point in our problem is unknown beforehand; and (\romannumeral 3) the shortest path length is unknown beforehand in this problem, whereas the maximum path length $l$ in our problem is given beforehand.    

Finally, our problem   shares at least a  common aspect(s)  with other visibility  problems \cite{EmoWelzl:constructing,DTLee:Computational,ARicci:AnAlgorithm,JeffreyFEastman:an,JosephORourk:art,WeipangChin:optimum,YongKHwang:gross} in this community, since they all  involve the concept of obstacles. But they are more or less different from our problem. The art gallery problem \cite{DTLee:Computational} for example is to find  a  minimum  set $\mathscr{S}$ of points such that for any point $q$ in a simple polygon $\mathscr P$, there is at least a point $p\in \mathscr S$ such that the segment $\overline{pq}$ does not pass through any edge of $\mathscr P$.  The edges of $\mathscr P$  here correspond to the obstacles.  The differences between this problem and our problem is obvious since our problem is to find a region rather than a  minimum  set  of points. More details about the visibility problems please refer to the textbooks, e.g., \cite{SubirKumarGhosh:visibility,MarkdeBerg:computational}.

\paragraph*{\textbf{\upshape Our contributions}}  
We  formulate the problem --- finding achievable region, offer insights  into its nature,  and develop multiple algorithms to tackle it. 

Specifically, in Section 3, we present a simpler-version    algorithm for the sake of intuition. The basic idea of this algorithm is to reduce our problem to computing the union of a set of {circular visibility region}s (CVRs), defined in Section \ref{sec:problem definition}. Intuitively, the CVR can be obtained by computing a boolean union of the visibility region  and the circle. The  visibility region however, may not be always bounded, which  leads to some troubles  and makes this straightforward idea to be difficult to develop. We adopt a different way to compute the CVR instead of  directly solving those troubles. Specifically, we first prune some unrelated obstacles  based on a  simple but efficient pruning mechanism that  ensures all candidate obstacles to be located in the circle, which can  simplify the subsequent computation.  After this, we use the idea of the  rotational plane sweep to construct the CVR, whose boundaries are represented as a series of vertexes and \textit{appendix point}s (defined in Section \ref{sec:problem definition}) that are stored using a double linked list.
We note that, for most CVRs, the circles used to  construct them  are unavailable beforehand. We use the    visibility  graph technique to  obtain the circle. Once we obtain a new CVR, we merge it with the previous one. In this way, we finally get our wanted answer, i.e., the achievable region $\mathscr R$.  This algorithm takes $O(n^3)$ time. 

We analyse the dominant steps of the first algorithm, and break through its bottleneck by incorporating the short path map (SPM)  technique. Specifically, in Section 4, we show the SPM  technique, previously used to answer the short path queries among polygonal obstacles, can be equivalently applied to the context of our concern. We use this technique to obtain those circles (that are used to  construct most CVRs).  Obtaining each circle needs $O(n^2)$ time in the first algorithm,  it takes (only)   $O(\log n)$ time by using the SPM technique. This improvement immediately yields an $O(n^2\log n)$ algorithm.

Thanking   the realization in Section 4 --- the SPM technique can be used to the context of our concern,  
the third algorithm presented in Section 5 also uses this technique. It however, does not continue to construct the CVRs. Instead, it directly traverse each region of the SPM to trace the boundaries. By doing so, it gets   a \textit{circular kernel-region} and many \textit{circular ordinary regions}. (Some circular ordinary
regions possibly consist of not only circular arc  and straight line segments but also hyperbolas, hence, sometimes we call them  \textit{conic polygons}.)  Any two of these regions actually have no \textit{duplicate region}, defined in Section \ref{sec:problem definition}. In theory, we can directly output these conic polygons. We should note that Section \ref{sec:problem definition} emphasizes a constraint condition --- the output of the algorithm to be developed is the well-organized boundaries of the achievable region $\mathscr R$.  In order to satisfy such a constraint, we only need to execute a simple \textit{boolean set operation} (i.e.,   arrange all  edges (segments) of these conic polygons and then combine them in order). Due to the SPM  has complexity $O(n)$, naturally,  the number of edges of all these conic polygons  has  the linear-size complexity. Moreover, in the context of our concern, the number of intersections among all these segments is clearly no more than  $O(n)$, and  constructing the SPM  can be done in $O(n\log n)$ time (which has been stated in Section \ref{sec:modified algorithm}), all these facts form our final algorithm, obtaining an  $O(n\log n)$ worst case upper bound.   

We remark that although this paper focuses on the case of line-segment obstacles,  the FAR problem in the case of polygonal obstacles should  be easily  solved using anyone of the above algorithms (maybe some minor modifications are needed).

\paragraph*{\textbf{\upshape Paper organization}}
In the next section, we  formulate our problem, define some notations, and analyse the so-called rough solution mentioned before. Section \ref{sec:solution} presents a simpler-version $O(n^3)$ algorithm for the sake of intuition. Section \ref{sec:modified algorithm} presents a modified-version  $O(n^2\log n)$ algorithm,  and Section \ref{sec:more efficient solution} presents our final algorithm, running in $O(n\log n)$ time.  Finally, Section \ref{sec:conclusion} concludes this paper with several open problems.

\section{Preliminaries} \label{sec:problem definition}
\subsection{Problem definition and notations}
Let $\mathscr M$ be a moving object in $\mathbb{R}^2$, we assume   $\mathscr M$ can be regarded as a point compared to the total space. Let $\mathscr O$ be a set of $n$ disjoint obstacles in $\mathbb{R}^2$. We assume  that the moving  
object  $\mathscr M$ can  freely move in $\mathbb{R}^2$, but cannot directly pass through any obstacle $o\in \mathscr O$. For clarity, we use $\mathbb{R}^2/ \mathscr O$ to denote the free space.  Let  $s$ be the starting point of the moving object $\mathscr M$. (Sometimes we also call it the \textit{source} point.) Given a point $p\in \mathbb{R}^2/ \mathscr O$, assume that the moving object $\mathscr M$ freely moves from $s$ to $p$,  the total travelled-distance of $\mathscr M$ is called  the \textit{path length}.  We remark that, if  $s$ and  $p$ are identical in $\mathbb{R}^2$, the \textit{path length} is not definitely equal to 0. The maximum value of path length that $\mathscr M$ is allowed to  travel is called the  \textit{maximum path length}.


We use $l$ to denote the maximum path length of the moving object $\mathscr M$. Given two points $p$ and $p^\prime$ in $\mathbb{R}^2$, we use $\pi(p,p^\prime)$ to denote the  shortest path length in presence of obstacles (a.k.a., the geodesic distance). When  two points $p$ and $p^\prime$ are to be visible to each other, we use $dist(p,p^\prime)$ to denote the Euclidean distance between them. Given a point $p$ in $\mathbb{R}^2$, we say  $\mathscr M$ can reach the point $p$ such that $\pi(s,p)\leq l$. (Sometimes, we also say $p$ is achievable for $\mathscr M$.) Given two closed regions, we say they have the  \textit{duplicate region} if the area of their intersection set does not equal to 0; otherwise, we say they have no duplicate region.

\begin{definition}[Finding achievable region]
The problem of {finding achievable region} (FAR) is to find a region denoted by $\mathscr R$ such that for any point $p$ in $ {\mathbb R}^2$, if $\mathscr M$ can reach  $p$, then $p\in \mathscr R$; otherwise, $p\notin \mathscr R$. 
\end{definition}



Clearly, the line segment is the  basic element in geometries. This paper restricts the attention to   the case of disjoint line-segment  obstacles. (We remark that the case of polygonal obstacles should be easily solved using  our proposed algorithms.)  In particular, we emphasize that the output of the algorithm to be developed is the well-organized boundaries of the achievable region $\mathscr R$, rather than a set of \textit{out-of-order} segments.  Furthermore, we are interested in developing the exact algorithms rather than approximate algorithms. Unless stated otherwise, the term \textit{obstacles} refers to \textit{line-segment obstacles} in the rest of the paper.

Let $\mathscr E$ be the set of all endpoints of obstacles $\mathscr O$, we use $\mathscr E^\prime$  to denote the set of endpoints  such that (\romannumeral 1) for each endpoint $e\in \mathscr E^\prime$, $\pi (s,e)<l$; and (\romannumeral 2) for each endpoint $e\in \mathscr E$ but $e\notin \mathscr E^\prime$, $\pi (s,e)\geq l$. For clarity, we call the set $\mathscr E^\prime$ the \textit{effective endpoint set}. (Sometimes we call the endpoint $e\in \mathscr E^\prime$ the \textit{effective endpoint}.)  Given any set, we use the notation ``$|\cdot|$'' to denote the cardinality of the set (e.g., $|\mathscr E^\prime|$ denotes the number of effective endpoints).   Given two points $p$ and $p^\prime$, we use $\overline{pp^\prime}$ to denote the straight line segment joining the two points. Given a circular arc segment with two endpoints $p$ and $p^\prime$, we use  $\widehat{pp^*p^\prime}$ to denote this arc. Note that, here $p^*$ is  used to eliminate the ambiguity. (Recall that a circular arc may be the major/minor arc, two endpoints cannot determine a specific circular arc.) We call such a point (like $p^*$) the  \textit{appendix point}. Given a circle with the center $p$ and the radius $r$, we use $C(p,r)$ to denote this circle. Given a  ray emitting
from a point $p$ and passing through another point $p^\prime$, we use $r_{(p,p^\prime)}$ to denote this ray.  Given two points $p$ and $p^\prime$,  we use $\sphericalangle (p,p^\prime)$ to denote that they are visible to each other, and use $\lnot(\sphericalangle (p,p^\prime)$ to denote the opposite case. Given a point $p$, we use $\mathscr R_{vp}(p)$ to denote the visibility polygon (i.e., visibility region) of $p$.




\subsection{A brief  analysis on the {``rough solution''} }


Now we  verify  whether or not the so-called rough solution (\textsf{RS})  can  work correctly.   
\begin{example}\label{example: analyse rs}
\textup{Based on the  \textsf{RS}, we first get  ${\mathscr R}_d$ (recall Section \ref{sec:1}); it here is bounded by $\overline{v_1a}$, $\overline{ab}$, $\overline{bv_2}$, $\widehat{v_2v_6v_4}$, $\overline{v_4d}$, $\overline{dc}$, $\overline{cv_3}$, $\widehat{v_3v_5v_1}$. See Figure \ref{fig:1:6}.  Next, we  get four circles which are ${C}(a,dist(a,v_1))$,  ${C}(b,dist(b,v_2))$, ${C}(c,dist(c,v_3))$, and ${C}(d,dist(d,v_4))$, respectively. We finally get the answer by merging all the five regions. Interestingly, this example  shows that the \textsf{RS} seems to be feasible.  The example in Figure \ref{fig:1:7}, however,  breaks this delusion.  See the black region in the circle ${C}(a,dist(a,v_1))$. $\mathscr M$ obviously can not reach  any point located  in this region. In other words, the second step of the \textsf{RS} is  incorrect. }   

\textup{Furthermore, consider the example in Figure \ref{fig:1:8}, here ${\mathscr R}_d$ is bounded by $\overline{v_1a}$, $\overline{ab}$, $\overline{bv_4}$, $\overline{v_4d}$, $\overline{dv_2}$, $\widehat{v_2v_3v_1}$. Note that,  $v_4$ here  is the intersection between $r_{({s},b)}$ and $\overline{cd}$, and $\pi (s,v_4)<l$.   Then, whether or not we should  consider  such a point like $v_4$ when we design a new solution? } 
\end{example}

In the next section, we show such a point does not need to be considered (see Lemma \ref{lemma:reduce to cvr}), and prove that  our problem can be reduced   to computing the union of a series of  \textit{circular visibility region}s defined below. 

\begin{definition}[Circular visibility region]\label{defintion:cvr}
Given  a circle $C(p,r)$,   its corresponding circular visibility region, denoted by ${\mathscr R}_{cvr}(p,r)$, is a region  such that, for any point $p^\prime\in {C}(p,r)$, if $p$ and $p^\prime$ are to be visible to each other, i.e., $\sphericalangle (p,p^\prime)$, then $p^\prime\in {\mathscr R}_{cvr}(p,r)$; otherwise, $p^\prime\notin {\mathscr R}_{cvr}(p,r)$.  
\end{definition}

We say $p$ is the center  and $r$ is the radius of ${\mathscr R}_{cvr}(p,r)$, respectively.  Sometimes, we also say that ${\mathscr R}_{cvr}(p,r)$ is the circular visibility region of $p$ when $r$ is clear from the context.  

\begin{figure}[h]
  \centering
  \subfigure[\scriptsize {  } ]{\label{fig:1:6}
     \includegraphics[scale=.36]{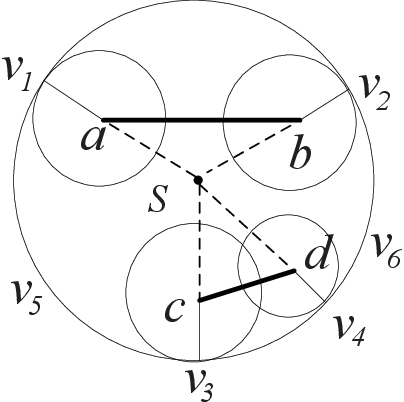}} 
  \subfigure[\scriptsize { }]{\label{fig:1:7}
     \includegraphics[scale=.36]{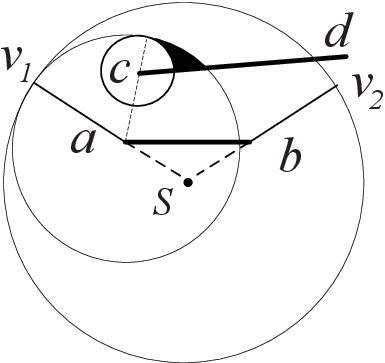}} 
  \subfigure[\scriptsize { }]{\label{fig:1:8}
    \includegraphics[scale=.36]{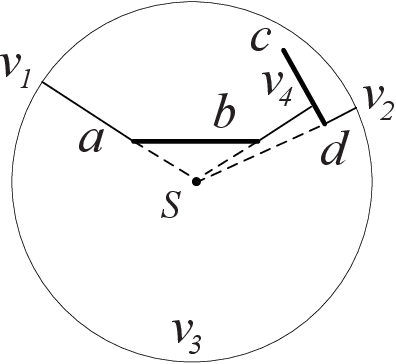}} 
 \caption{\small The case of two obstacles. $\overline{ab}$ and $\overline{cd}$ denote  obstacles, other solid  or dashed lines are the auxiliary lines,  for ease of presentation.} 
 \label{fig:1aaa}
\end{figure}



\section{An  $O(n^3)$ algorithm}\label{sec:solution}

\subsection{Reduction}\label{sec:reduction}

\begin{lemma}\label{lemma:reduce to cvr}
Given a set $\mathscr O$ of  disjoint obstacles, the moving object $\mathscr M$, its starting point ${s}$ and its maximum path length $l$,  the effective endpoint set $\mathscr E^\prime$, and a point $p$ $\in C(s,l)$, we have that if $\pi (s,p)<l$,  then ${\mathscr R}_{cvr}(p, l-\pi(s,p))\subset \bigcup_{e\in \mathscr E^\prime} {\mathscr R}_{cvr}(e, l-\pi (s,e))$ $\bigcup {\mathscr R}_{cvr}(s, l)$.  
\end{lemma}
\begin{proof} 
The proof is not difficult but somewhat long, we move it to  Appendix A. 
\end{proof}


Based on Lemma \ref{lemma:reduce to cvr}, the \textit{set theory} and  the definition of the circular visibility region, we can easily build the follow theorem. 

\begin{theorem}\label{theorem:reduce}
Given a set $\mathscr O$ of  disjoint obstacles, the moving object $\mathscr M$, its starting point ${s}$ and its maximum path length $l$, and the effective endpoint set $\mathscr E^\prime$; 
then, the achievable region $\mathscr R$ can be computed as ${\mathscr R}=\bigcup_{e\in \mathscr E^\prime}{\mathscr R}_{cvr}(e, l- \pi (s,e))$ $\bigcup {\mathscr R}_{cvr}(s, l)$. 
\end{theorem}

Theorem \ref{theorem:reduce} implies that the achievable region $\mathscr R$ is a union of a set of circular visibility regions (CVRs). Clearly,  the boundaries of the CVR usually  consist of  circular arc and/or straight line   segments. The straight line segment can be represented using two endpoints, and the circular arc can be represented using two endpoints and an \textit{appendix point} (recall Section \ref{sec:problem definition}). Naturally, a double linked list can be used to store the information of  boundaries of a CVR. Meanwhile, we can easily know that the boundaries of  $\mathscr R$ also consist of circular arc and/or straight line segments, since it is a union of all the CVRs, implying that we can also use a double linked list to store the information of   boundaries of  $\mathscr R$.


\noindent \textbf{Discussion.}
Given a point $p$ and a positive real number $r$, an easy called to mind method to compute the circular visibility region of $p$ is as follows. First, we compute  ${\mathscr R}_{vp}(p)$ (i.e., the visibility polygon of $p$, recall Section \ref{sec:problem definition}). Second, we compute the intersection set  between the circle ${C}(p,r)$ and the visibility polygon ${\mathscr R}_{vp}(p)$. Clearly, we  can easily get all the  segments that are visible from $p$,  based on existing algorithms \cite{subhansh:worstcase,TakaoAsano:Visibility}. However, the visibility polygon of a point among  a set of obstacles  may not be always bounded \cite{SubirKumarGhosh:visibility}, which arises the following trouble --- 
how to compute the intersection set between an unbounded polygonal region and a circle? 
In the next section, we present a   method  that can compute the  circular visibility region of $p$ efficiently, and is to be free from  tackling the above trouble.

\subsection{Computing the circular visibility region}

\subsubsection{Pruning}\label{subsec:prun}

Before constructing the CVR, we first  prune some unrelated obstacles using a simple but efficient mechanism, which can  simplify the subsequent computation.

\begin{lemma}\label{lemma:pruning}
{
Given the set $\mathscr O$ of obstacles and a circle ${C}(p,r)$, we have that for any obstacle $o\in \mathscr O$, 
\begin{itemize*}
\item if $o\notin$${C}(p,r)$; then, $o$ can be  pruned safely.
\item if  $o\cap$${C}(p,r)$, let $o_s$ be the sub-segment of $o$ such that $o_s\notin C(p,r)$;  then, the sub-segment $o_s$  can be   pruned safely. 
\end{itemize*}
}
\end{lemma}
\begin{proof}
The proof is straightforward. We only need to prove that $o$  makes no impact on the size of  ${\mathscr R}_{cvr}(p,r)$. We prove this by contradiction.  According to the definition of the CVR, we can easily know that for any point $p^\prime \in {\mathscr R}_{cvr}(p,r)$, the point $p^\prime$ has the following properties: $p^\prime\in {C}(p,r)$ and $\sphericalangle(p,p^\prime)$. We assume that $o$ makes impact on the size of ${\mathscr R}_{cvr}(p,r)$. This implies that   there is at least a point $p^{\prime\prime} \in o$  such that $dist(p,p^{\prime\prime})<r$. On the other hand, based on analytic geometry, it is easy to know that  if $o\notin {C}(p,r)$, then, for any point $p{\prime\prime}\in o$,  $dist(p,p^{\prime\prime})>r$. It is contrary to the previous assumption. Hence $o$ can be pruned safely. (Using the same method, we can prove that the sub-segment $o_s$  can also be pruned safely.)
\end{proof}



\begin{theorem}\label{theorem:prun}
Given  a set  of $n$ obstacles in ${\mathbb{R}^2}$, and a circle ${C}(p,r)$, we can prune the unrelated obstacles  in $O(n)$ time.
\end{theorem}
\begin{proof}
We first use the {minimum bounding rectangle} (MBR) of ${C}(p,r)$ as the input to prune unrelated obstacles. As a result,  we  obtain a set of obstacles  called   \textit{initial candidate obstacle}s (ICOs), see Figure \ref{fig:3n:a}. This step takes  $O(n)$ time.  Next, for each ICO, we check if it can be pruned or partially pruned according to Lemma \ref{lemma:pruning}. Each operation takes constant time and the number of ICOs is $\Omega(n)$ in the worst case. Thus, the total running time for pruning obstacles is $O(n)$.
\end{proof}

\begin{figure}[h]
  \centering
  \subfigure[\scriptsize {  } ]{\label{fig:3n:a}
     \includegraphics[scale=.67]{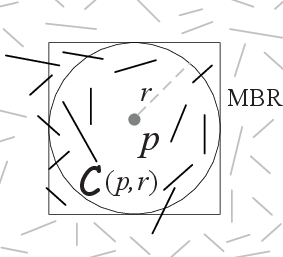}} \hspace{4ex}
  \subfigure[\scriptsize { }]{\label{fig:3n:b}
     \includegraphics[scale=.583]{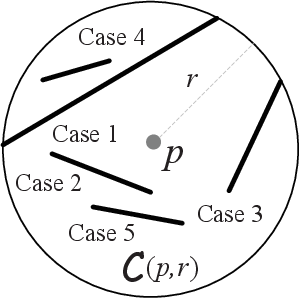}} 
 \caption{\small Example of pruing obstacles. (a) The rectangle denotes the MBR of ${C}(p,r)$, and the black solid line segments denote initial candidate obstacles. (b) The black solid line segments denote candidate obstacles. } 
 \label{fig:3}
\end{figure}

We say  the rest of obstacles (after pruning) are  \textit{candidate obstacle}s. In particular, they  can be generally classified into five types, see Figure \ref{fig:3n:b}. 
We remark  that it is also feasible to develop more complicated pruning mechanism that can prune the Case 4-type obstacle. 

In particular, we can easily realize an obvious characteristic --- all these candidate obstacles are bounded by the circle ${C}(p,r)$.  This fact lets us be without distraction from handling many unrelated obstacles, and be easier to construct the CVR.  In the next subsection, we show how to construct it and store its boundaries in the aforementioned (well-known) data structure ---  double linked list. 

\subsubsection{Constructing}\label{subsec:construction}
The basic idea of constructing the CVR is based on the rotational plane sweep \cite{subhansh:worstcase}. We first introduce some basic concepts, for ease of understanding the detailed process of constructing the CVR. 

Given the circle $C(p,r)$, we use $L$  to denote a horizontal ray emitting from $p$ to the right of $p$. When we rotate $L$ around $p$, we call the obstacles that currently intersect with $L$ the \textit{active obstacles}. 
Assume that $L$ intersects with an obstacle at $u$, we use $o(u)$ to denote this obstacle. Let $T$ be a balance tree, we shall insert the active obstacles into $T$. Assume that $o(e_i)$ is a node in $T$, we use  $o(e_j)$ to denote the parent node of $o(e_i)$ in $T$  by default.     Given an obstacle endpoint $e$ and the horizontal ray $L$, we use $\angle(e)$ to denote the  counter-clockwise angle subtended by the segment $\overline{pe}$ and the ray ${L}$. We say this angle is the \textit{polar angle} of $e$. Given  two endpoints $e$ and $e^\prime$ of an obstacle $o$, if $\angle(e)>\angle (e^\prime)$, we say $e$ is the \textit{upper endpoint} and $e\prime$ is the \textit{lower endpoint}; and vice verse.   
Moreover, let $\mathscr D$ be a double linked list used to store the CVR. 

We remark that   two endpoints of an obstacle $o$ can have the same size of polar angles; this  problem can be  solved by directly ignoring $o$, since this obstacle $o$ almost makes no impact  on the size of ${\mathscr R}_{cvr}(p,r)$. In the subsequent discussion,  we assume  two endpoints of any obstacle $o$ have different  polar angles.

\begin{theorem}\label{theorem:constructing}
Given a circle ${C}(p,r)$,  without loss of generality, assume that there are $n^\prime (\leq n)$ candidate obstacles; then constructing ${\mathscr R}_{cvr}({p},{r})$ can be finished in  $O(n\log n)$ time. 
\end{theorem}

\begin{proof}
We first draw the horizontal ray ${L}$ from  $p$ to the right of $p$, and then compute the intersections between  ${L}$ and all candidate obstacles, which takes $O(n^\prime)$ time in the worst case.  
Without loss of generality, assume that there are $k$ intersections. Let $u_i$ denote one of these  intersections,  we sort these intersections  such that $dist(p,u_{i-1})<dist(p,u_i)<dist(p,u_{i+1})$ for any $i\in [2,3,$ $\cdots,k-1]$. (Note that, these active obstacles are also being sorted when we sort the intersections.) This takes $O(n^\prime \log n^\prime)$ time, sine $k$ can have $\Omega(n^\prime)$ size in the worst case, see Figure \ref{fig:4n:a}.

We then  insert these active obstacles ${o}(u_i)$ ($i\in [1,\cdots,k]$) into the balance tree $T$ such that,  if $o(u_i)$ is the left (right) child of $o(u_j)$, then $i<j$ ($i>j$). So, $o(u_1)$ is the leftmost leaf of $T$. (Note that, any point $p^*  \in o(u_i)$ is visible from $p$ if and only if $o(u_i)$ is the leftmost leaf in $T$.)   Initializing the tree $T$ takes $O(n^\prime)$ time, since $|T|=k=\Omega(n^\prime)$ in the worst case.

We sort the endpoints of  all candidate obstacles according to their  polar angles.  We denote these sorted endpoints as  $e_1$, $e_2$, $\cdots$, $e_{2n^\prime}$  such that $\angle (e_{i-1})<$  $\angle (e_i)<$  $\angle(e_{i+1})$ for any $i\in [2,3,\cdots,2n^\prime-1]$, see Figure \ref{fig:4n:b}.   Since each obstacle has two endpoints, sorting them takes $O(n^\prime \log n^\prime)$ time. We then rotate  ${L}$ according to the counter-clockwise. Note that, in the process of rotation, the active obstacles in $T$ change if and only if $L$ passes through the endpoints of obstacles. We handle each endpoint event as follows.

{\vspace{.5ex} 
 \hrule
\vspace{.5ex}

}

    { \small

1$~~~~$\textbf{if} $e_i$ is the lower endpoint of obstacle $o(e_i)$  \textbf{then}

2$~~~~$$~~~~$Insert the obstacle $o(e_i)$  into $T$ 

3$~~~~$$~~~~$\textbf{if} $o(e_i)$ is the leftmost leaf in $T$ \textbf{then}

4$~~~~$$~~~~$$~~~~$\textbf{if}  $|T|>1$ \textbf{then}

5$~~~~$$~~~~$$~~~~$$~~~~$Let $z$ be  the intersection  between $\overline{pe_i}$ and $o(e_j)~$ 

$~~~~$$~~~~$$~~~~$$~~~~~~$Put the two points $z$ and $e_i$ into ${\mathscr D}$ in order

6$~~~~$$~~~~$$~~~~$\textbf{else if} $|T|=1$ $\land$ $dist(p,e_i)\neq r$ \textbf{then}

7$~~~~$$~~~~$$~~~~$$~~~~$Let $z$ be  the intersection  between $\overline{pe_i}$ and ${ C}(p,r)$

$~~~~$$~~~~$$~~~~$$~~~~~~$Put the two points $z$ and $e_i$ into ${\mathscr D}$ in order

8$~~~~$$~~~~$$~~~~$\textbf{else} // $|T|=1$ $\land$ $dist(p,e_i)=r$

9$~~~~$$~~~~$$~~~~$$~~~~$Put the point $e_i$ into ${\mathscr D}$ 

10$~~~$\textbf{if} $e_i$ is the upper endpoint  of  $o(e_i)$  \textbf{then}

11$~~~~$$~~~$\textbf{if} $o(e_i)$ is the leftmost leaf in $T$

12$~~~~$$~~~~$$~~~$\textbf{if}  $|T|>1$ \textbf{then}

13$~~~~$$~~~~$$~~~~$$~~~$Let $z$ be  the intersection  between $\overline{pe_i}$ and $o(e_j)$

$~~~~$$~~~~$$~~~~$$~~~~~~~$Put the two points  $e_i$ and $z$  into ${\mathscr D}$ in order

14$~~~~$$~~~~$$~~~$\textbf{else if} $|T|=1$ $\land$ $dist(p,e_i)\neq r$ \textbf{then}

15$~~~~$$~~~~$$~~~~$$~~~$Let $z$ be  the intersection  between $\overline{pe_i}$ and ${ C}(p,r)$

$~~~~$$~~~~$$~~~~$$~~~~~~~$Let $z^*$ be a point  such that $dist(p,z^*)=r$ and $\overline{pz^*}$ is 

$~~~~$$~~~~$$~~~~$$~~~~~~~$subtended by $\overline{pe_i}$ and $\overline{pe_{i+1}}$

$~~~~$$~~~~$$~~~~$$~~~~~~~$Put the three points  $e_i$, $z$ and $z^{*}$  into ${\mathscr D}$ in order

16$~~~~$$~~~~$$~~~$\textbf{else} // $|T|=1$ $\land$ $dist(p,e_i)=r$

17$~~~~$$~~~~$$~~~~$$~~~$Let $z^*$ be a point  such that $dist(p,z^*)=r$ and $\overline{pz^*}$ is 

$~~~~$$~~~~$$~~~~$$~~~~~~~$subtended by $\overline{pe_i}$ and $\overline{pe_{i+1}}$

$~~~~$$~~~~$$~~~~$$~~~~~~~$Put the two points $e_i$ and $z^*$ into ${\mathscr D}$ in order

18$~~~~$$~~~$Remove  $o(e_i)$ from $T$

\vspace{.5ex}  
\hrule
\vspace{.5ex} 
        }

It is easy to know that the height of $T$ is $\log n^\prime$. So, each operation (e.g., insert, delete and find the  active obstacle)  in $T$   takes $O(\log n^\prime)$ time. In addition, other operation (e.g., obtain the intersections $z$, $z^*$, and  put them into ${\mathscr D}$) takes constant time. Thus, handling all endpoint events  takes $n^\prime \log n^\prime$ time. After all endpoints are handled, we finally get the CVR whose boundaries are represented as a series of vertexes and appendix points (stored in ${\mathscr D}$). See Figure \ref{fig:4n:b} for example,  we shall get a series of organized points  $\{e_2,z_1,z_1^*,$ $z_2,e_3,e_5,$ $z_3,e_6,z_2^*,$ $z_4,e_7,$ $z_5,e_8 \}$, where $z_1^*$ and $z_2^*$ are  appendix points.  

We note that   $n^\prime$ can have $\Omega(n)$ size in the worst case. In summary, constructing the CVR takes $O(n\log n)$ time. 
\end{proof}

\begin{figure}[t]
  \centering
  \subfigure[\scriptsize {  } ]{\label{fig:4n:a}
     \includegraphics[scale=.35]{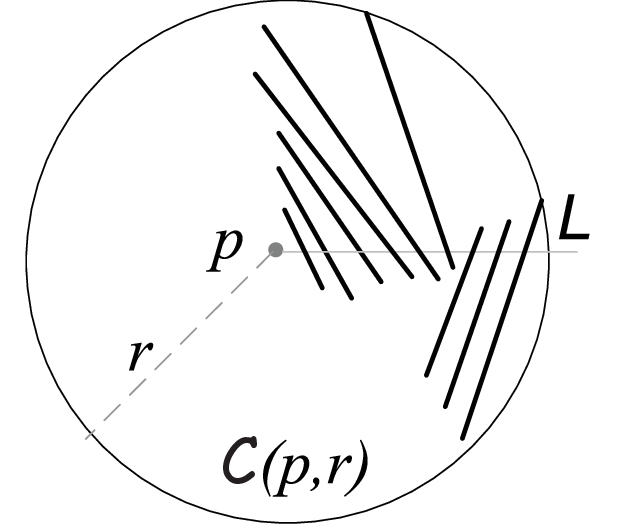}} \hspace{2ex}
  \subfigure[\scriptsize { }]{\label{fig:4n:b}
     \includegraphics[scale=.35]{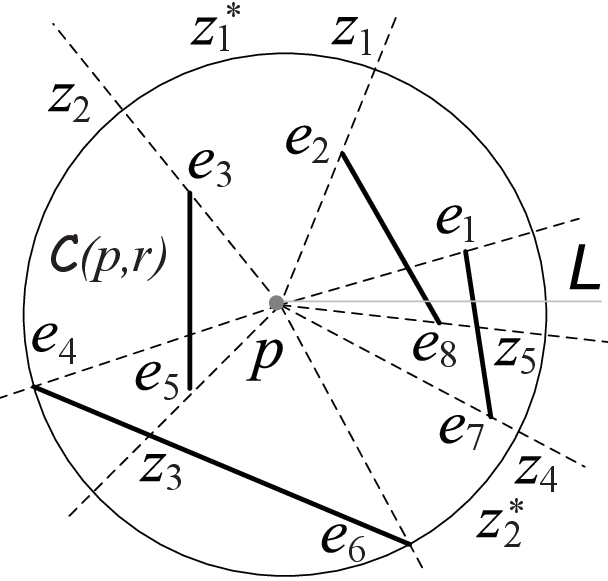}} 
 \caption{\small Example of constructing ${\mathscr R}_{cvr}(p,r)$. The  black and grey solid line segment(s) denote candidate obstacles and the horizontal ray $L$, respectively. (a) $r$ denotes the radius of ${ C}(p,r)$. (b) The dashed lines are the result of rotating $ L$.} 
 \label{fig:4}
\end{figure}

\noindent \textbf{Discussion.} Recall Section \ref{sec:reduction}, we mentioned two types of CVRs: ${\mathscr R}_{cvr}(s,l)$  and ${\mathscr R}_{cvr}(e, l-\pi(s,e))$. Regarding to the former, we use the circle ${ C}(s,l)$ to prune  unrelated obstacles and then to construct it based on the method discussed just now.  Regarding to the latter, the circle ${ C}(e, l-\pi(s,e))$ however, is unavailable beforehand. Hence,  we first have to obtain this circle. In the next subsection, we  show how to get it using a simple method. 
\subsubsection{Obtaining the  circle ${ C}(e, l-\pi(s,e))$}\label{subsec:two types of cvrs}
This simple method  incorporates the classical \textit{visibility graph} technique. To save space, we  simply state the previous results, and then use them to help us obtain this circle.  

\begin{definition}[Visibility graph]
\cite{EmoWelzl:constructing}
The visibility graph of a set of $n$ line segments is an undirected graph whose vertexes consist of all the endpoints of the $n$ line segments, and whose edges connect mutually visible  endpoints.  
\end{definition}

\begin{lemma}\label{lemma:construct visiblity graph}
\cite{EmoWelzl:constructing} Given a set of $n$ line segments in $\mathbb{R}^2$, constructing the visibility graph for these segments can be finished in  $O(n^2)$ time.
\end{lemma}

\begin{lemma}
\cite{thomasHCormen:introduction} \label{lemma:standard dijkstra}
Given an undirected graph with $n$ vertexes, finding the shortest path between any pair of vertexes can be finished in $O(n^2)$ time using the standard Dijkstra  algorithm.
\end{lemma}

To obtain the circle, our first step is to construct a visibility graph. Note that, we here need to consider the starting point $s$; in other words, the visibility graph consists of not only endpoints of  obstacles but also the starting point $s$. Even so, from Lemma \ref{lemma:construct visiblity graph}, we can still get an immediate corollary below.  

\begin{corollary}\label{corollary visibility graph}
Given a set $\mathscr O$ of $n$ disjoint line-segment obstacles, and the starting point $s$, we can build a visibility graph for the set of obstacles and the starting point $s$ in $O(n^2)$ time.
\end{corollary}

Let $\mathscr G$ be the visibility graph obtained using the above method. 
For clarity, we say $C(e,l-\pi(s,e))$ is a \textit{valid circle} if $l-\pi(s,e)>0$; otherwise, we say it is an \textit{invalid circle}. 


\begin{theorem}\label{corollary:shortest path}
Given the visibility graph $\mathscr G$, the maximum path length $l$, the staring point $s$, and  an obstacle endpoint $e$,   obtaining the circle ${ C}(e, l-\pi(s,e))$  can be finished in $O(n^2)$ time. 
\end{theorem}
\begin{proof}
It is easy to see that the key step of obtaining the circle $C(e,l-\pi(s,e))$ is to compute  the shortest path length from $s$ to $e$, i.e.,  $\pi(s,e)$. Now assume we have obtained the visibility graph $\mathscr G$. 
At anytime we need to obtain the circle ${ C}(e,l-\pi(s,e))$, we run the \textit{standard  Dijkstra  algorithm}  to find the shortest path from $s$ to   $e$, which  takes $O(n^2)$ time{\small \footnote{ \small We also note that Ghosh and Mount \cite{subirkumargosh:anoutput} proposed an output-sensitive $O(E+n\log n)$ algorithm for constructing the visibility graph, where $E$ is the number of edges in the graph. Furthermore,  using Fibonacci heap,  the shortest path of two points in a graph    can be reported in $O(E+n\log n)$ time \cite{MichaelLFredman:fibonacci}. Even so, the worst case running time is still no better than $O(n^2)$ since the visibility graph can have $\Omega (n^2)$ edges in the worst case. }}, see Lemma \ref{lemma:standard dijkstra}. Once the shortest path is found, its length $\pi(s,e)$ can be easily computed in additional $O(k)$ time, where $k$ is the number of segments in the path. Finally,  if $l-\pi(s,e)>0$, we let $l-\pi(s,e)$ and $e$  be  the radius and center of the  circle, respectively, we get a \textit{valid circle}; otherwise, we report it is an \textit{invalid circle}. This can be finished in $O(1)$ time. Pulling all together, this completes the proof. 
\end{proof}

\subsection{Putting it all together}\label{subsec:the overall algorithm}
The  overall algorithm  is shown in Algorithm \ref{alg:overall algorithm}.  The correctness of our algorithm follows from Lemma \ref{lemma:pruning}, Corollary \ref{corollary visibility graph}, Theorems  \ref{theorem:reduce},  \ref{theorem:constructing} and  \ref{corollary:shortest path}. 

\begin{theorem}\label{theorem:algorithm 1 running time}
The running time of Algorithm \ref{alg:overall algorithm} is $O(n^3)$. 
\end{theorem}
\begin{proof}
Clearly, Lines 2, 3 and 4 take $O(n^2)$, linear and $O(n\log n)$ time, respectively, see Corollary \ref{corollary visibility graph}, Theorems \ref{theorem:prun} and  \ref{theorem:constructing}. Within the \textbf{for} circulation, Line 6 takes $O(n^2)$ time, see Theorem \ref{corollary:shortest path}. Line 8 takes linear time, see Theorem \ref{theorem:prun}. Line 9 takes $O(n\log n)$ time, see Theorem \ref{theorem:constructing}. We remark  that the step ``let $\mathscr R=\mathscr R\bigcup {\mathscr R}_{cvr}(e,l-\pi(s,e))$'' shown in Line 8 is a simple boolean union operation of two polygons with circular arcs, it is used to remove the \textit{duplicate region}. A straightforward adaptation of Bentley-Ottmann's  plane sweep algorithm \cite{JonLouisBentley:Algorithms}, or  the algorithm in \cite{EricBerberich:aComputational} can be used to obtain their union in $O((m+k)\log m)$ time, where $m$ and $k$ respectively are the number of edges and  intersections of the two polygons. Regarding to  the case of our concern, $m=\Omega(n)$ and $k$ has the constant descriptive complexity, see e.g., Figure \ref{fig:4n:b} for an illustration (note: substitute $C(p,r)$ with $C(s,l)$). Hence, the step ``let $\mathscr R=\mathscr R\bigcup {\mathscr R}_{cvr}(e,l-\pi(s,e))$'' actually can be done in $O(n\log n)$ time.  Hence, the \textbf{for} circulation takes $O(n^3)$ time. To summarize, the worst case upper bound of this algorithm is $O(n^3)$. 
\end{proof}

\begin{algorithm}[h]
\caption{ {Finding Achievable Region of $\mathscr M$}} 
\label{alg:overall algorithm} 
\begin{algorithmic}[1] 
\REQUIRE  $\mathscr O$, $s$, $l$
\ENSURE $\mathscr R$
    { \small
        \STATE  Set $\mathscr R=\emptyset$
        \STATE  Construct the visibility graph $\mathscr G$
        \STATE  Prune unrelated obstacles using ${C}(s,l)$
        \STATE  Construct  ${\mathscr R}_{cvr}(s,l)$, and let $\mathscr R={\mathscr R}_{cvr}(s,l)$
        \FOR {each obstacle endpoint $e$}
            \STATE Obtain the circle ${C}(e, l-\pi(s,e))$
            \IF{it is  a \textit{valid circle}} 
                \STATE Prune unrelated obstacles using ${C}(e, l-\pi(s,e))$ 
                \STATE Construct  ${\mathscr R}_{cvr}(e,l-\pi(s,e))$, and let $\mathscr R=\mathscr R\bigcup {\mathscr R}_{cvr}(e,l-\pi(s,e))$
            \ENDIF
        \ENDFOR 
        \RETURN $\mathscr R$
        }
    \end{algorithmic}
\end{algorithm}

%



\noindent \textbf{Summary}
In this section, we have presented a simpler-version algorithm, which is indeed intuitive and easy-to-understand. We can easily see that the dominant step of this algorithm is to  obtain the circle ${ C}(e,l-\pi(s,e))$, i.e., Line 6.   In the next section, we show how to break through this bottleneck and obtain an $O(n^2\log n)$ algorithm.  

\section{An  $O(n^2\log n)$ Algorithm}\label{sec:modified algorithm}
This more efficient solution mechanically relies on the well-known technique called the \textit{shortest path map}, which was previously used to compute the Euclidean shortest path  among \textit{polygonal obstacles}. 

\subsection{Overview of the short path map}\label{subsec:overview of short path map}
\begin{definition}[Shortest path map] \cite{JosephSBMitchell:shortest,JohnHershberger:anoptimal}
The shortest path map of a source point $s$ with respect to a set ${\mathscr O}$ of obstacles is  a decomposition of the free space $\mathbb{R}^2\backslash$${\mathscr O}$ into regions, such that the shortest paths in the free space from $s$ to any two points in the same region pass through the same sequence of obstacle vertices. 
\end{definition}

The shortest path map is usually stored using the \textit{quad-edge data structure} \cite{Leonida:primitives,JosephSBMitchell:ANewAlgorithm,JohnHershberger:anoptimal}. Let $SPM(s)$ denote the short path map of the source point $s$.  It has the following properties. 
\begin{lemma}\label{lemma:point location query} 
\cite{JosephSBMitchell:shortest,JosephSBMitchell:theweighted,JosephSBMitchell:shortestpaoitp96}
Once the $SPM(s)$ is obtained, the map can be used to answer the single-source Euclidean shortest path  query in $O(\log n)$ time.
\end{lemma}

\begin{lemma}\label{lemma:spm has linear complexity}
\cite{JosephSBMitchell:shortest,JohnHershberger:anoptimal}
The map $SPM(s)$ has complexity $O(n)$, it consists of $O(n)$ vertexes, edges, and faces. Each edge is a segment of a line or a hyperbola. 
\end{lemma}

The early method to compute $SPM(s)$ can be found in \cite{JosephSBMitchell:ANewAlgorithm}, the author (Mitchell) later  adopted the \textit{continuous Dijkstra paradigm}  to compute this map \cite{JosephSBMitchell:shortest,JosephSBMitchell:shortestpaoitp96}. 
An  optimal algorithm for computing the Euclidean shortest path among a set of \textit{polygonal obstacles} was proposed by Hershberger and   Suri  \cite{JohnHershberger:anoptimal}, their method also used the continuous Dijkstra paradigm, but  it employed two key ideas: a \textit{conforming subdivision}  of the plane and an \textit{approximate wavefront}. Here we  simply state the general steps of constructing $SPM(s)$, and  their main result. (If any question, please refer to  \cite{JohnHershberger:anoptimal} for more details).  

The general steps of constructing $SPM(s)$ can be summarized as follows.
\begin{itemize*}
\item It builds a   confirming subdivision of the plane by considering only the vertexes of polygonal obstacles, dividing the plane into the linear-size cells.
\item It inserts the  edges of obstacles into the subdivision above, and  gets a   confirming subdivision  of the \textit{free space}.  
\item It propagates the approximate wavefront through the cells of the conforming subdivision of the free space, remembering the collisions arose from \textit{wavefront-wavefront} events and \textit{wavefront-obstacle} events.
\item It collects all the collision information, and uses them to determine all  the hyperbola arcs of $SPM(s)$, and finally combines these arcs with the edges of obstacles, forming $SPM(s)$.  
\end{itemize*}

\begin{lemma}\label{lemma:spm construction}\cite{JohnHershberger:anoptimal}
Given a source point $s$, and a set $\mathscr O$ of polygonal obstacles with a total number $n$ of  vertexes in the plane, the map $SPM(s)$  can be computed in $O(n\log n)$ time. 
\end{lemma}

\subsection{Constructing $SPM(s)$ among line-segment obstacles}
The method to construct $SPM(s)$ among line-segment obstacles is the same as the one in \cite{JohnHershberger:anoptimal}.  

To justify this, we   can consider  the line-segment obstacle as the special (or degenerate) case of the polygonal obstacle --- one has only 2 sides and no area (see Figure \ref{fig:5aN} for an illustration).  Furthermore,  although the free space in the case of line-segment obstacles is (almost) equal to the space of the plane (since each obstacle here has  no area), this  fact still cannot against  applying  the algorithm in \cite{JohnHershberger:anoptimal} to the case of our concern. This is mainly because (\romannumeral 1) we  can still  build the confirming subdivision of the plane by considering only the \underline{endpoints} of line segments firstly, and then  insert the $n$ \underline{line segments} into  the conforming subdivision; (\romannumeral 2)  the collisions are also arose from \textit{wavefront-wavefront} events and \textit{wavefront-obstacle} events, we can also collect these collision information, and then determine the hyperbola arcs of $SPM(s)$; and (\romannumeral 3) we can obtain $SPM(s)$ by (also) combining these arcs with the $n$ \underline{line segments}.  With the argument above, and each line-segment obstacle has only  $2$ endpoints,  from Lemma \ref{lemma:spm construction}, we have an immediate corollary below.
\begin{corollary}\label{corollary:short path map}
Given a source point $s$, and a set $\mathscr O$ of $n$ line-segment obstacles in the plane, the map $SPM(s)$  can be computed in $O(n\log n)$ time. 
\end{corollary}

\begin{figure}[h]
  \centering
     \includegraphics[scale=.52]{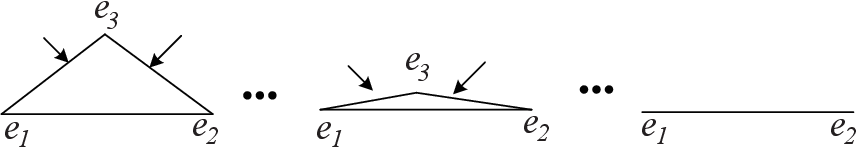} 
 \caption{\small The line-segment obstacle is the degenerate case of the polygonal obstacle. } 
 \label{fig:5aN}
\end{figure}

\subsection{The  algorithm}
To obtain the achievable region $\mathscr R$, the first step of this $O(n^2\log n)$ algorithm is to construct $SPM(s)$. The rest of steps are the same as the ones in Algorithm \ref{alg:overall algorithm} except the step ``obtain the circle $C(e,l-\pi(s,e))$''. We now can  obtain this circle in a more efficient way. This is mainly because the short path  can be computed  in $O(\log n)$ time once $SPM(s)$ is available, see Lemma \ref{lemma:point location query}. Based on this fact and the previous analysis used to prove Theorem \ref{corollary:shortest path}, we can easily build the following theorem.

\begin{theorem}\label{theorem:spm to get circle}
Given the short path map $SPM(s)$, the maximum path length $l$, and  an obstacle endpoint $e$,   obtaining the circle ${ C}(e, l-\pi(s,e))$  can be finished in $O(\log n)$ time.  
\end{theorem}

The correct of this algorithm follows from  Algorithm \ref{alg:overall algorithm},  Corollary \ref{corollary:short path map} and Theorem \ref{theorem:spm to get circle}.  The pseudo codes are shown in Algorithm \ref{alg:overall algorithm 2}.

\begin{algorithm}[h]
\caption{ {Finding Achievable Region of $\mathscr M$}} 
\label{alg:overall algorithm 2} 
\begin{algorithmic}[1] 
\REQUIRE  $\mathscr O$, $s$, $l$
\ENSURE $\mathscr R$
    { \small
        \STATE  Set $\mathscr R=\emptyset$
        \STATE  Construct $SPM(s)$
        \STATE  Prune unrelated obstacles using ${C}(s,l)$
        \STATE  Construct  ${\mathscr R}_{cvr}(s,l)$, and let $\mathscr R={\mathscr R}_{cvr}(s,l)$
        \FOR {each obstacle endpoint $e$}
            \STATE Obtain the circle ${C}(e, l-\pi(s,e))$ based on $SPM(s)$
            \IF{it is  a \textit{valid circle}} 
                \STATE Prune unrelated obstacles using ${C}(e, l-\pi(s,e))$ 
                \STATE Construct  ${\mathscr R}_{cvr}(e,l-\pi(s,e))$, and let $\mathscr R=\mathscr R\bigcup {\mathscr R}_{cvr}(e,l-\pi(s,e))$
            \ENDIF
        \ENDFOR 
        \RETURN $\mathscr R$
        }
    \end{algorithmic}
\end{algorithm}

\begin{theorem}\label{theorem:algorithm 2 running time}
The running time of Algorithm \ref{alg:overall algorithm 2} is $O(n^2\log n)$.
\end{theorem}
\begin{proof}
This follows directly from  Theorem \ref{theorem:spm to get circle} and the proof for Theorem \ref{theorem:algorithm 1 running time}.
\end{proof}

\noindent \textbf{Summary.} This section presented an $O(n^2\log n)$ algorithm by modifying Algorithm \ref{alg:overall algorithm}. We can easily see that the dominant step has shifted, compared to Algorithm \ref{alg:overall algorithm}. Now, The bottleneck is located in Line 9, which takes $O(n\log n)$ time.  In the next section, we show how to improve this $O(n^2\log n)$ algorithm to obtain a sub-quadratic algorithm. 

\section{An  $O(n\log n)$ Algorithm}\label{sec:more efficient solution}
The first step of this $O(n\log n)$ algorithm is also to construct the short path map, which is the same as the one in Algorithm \ref{alg:overall algorithm 2}. It however, does not construct the CVRs. Instead, it directly traverses each region of the short path map to obtain their boundaries,  and finally merges them.

\begin{figure*}[t]
  \centering
  \subfigure[\scriptsize {  } ]{\label{fig:6n:a}
     \includegraphics[scale=.36]{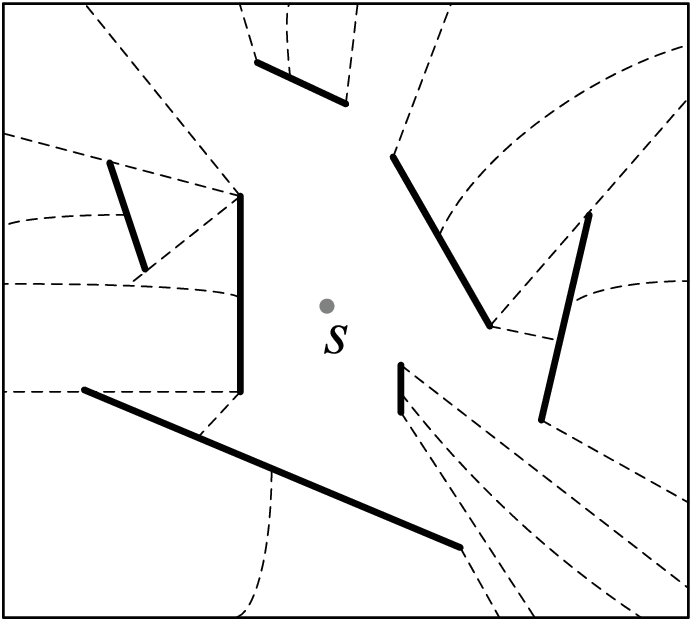}} 
  \subfigure[\scriptsize { }]{\label{fig:6n:b}
     \includegraphics[scale=.36]{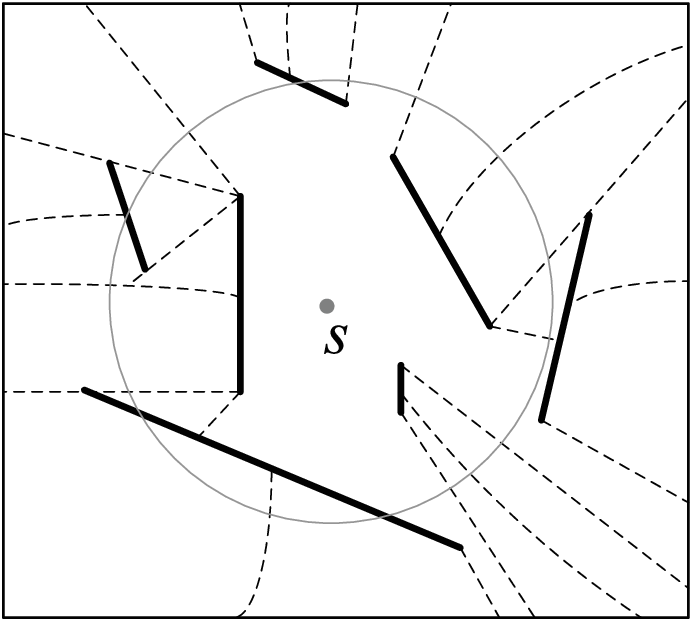}} 
  \subfigure[\scriptsize { }]{\label{fig:6n:c}
    \includegraphics[scale=.36]{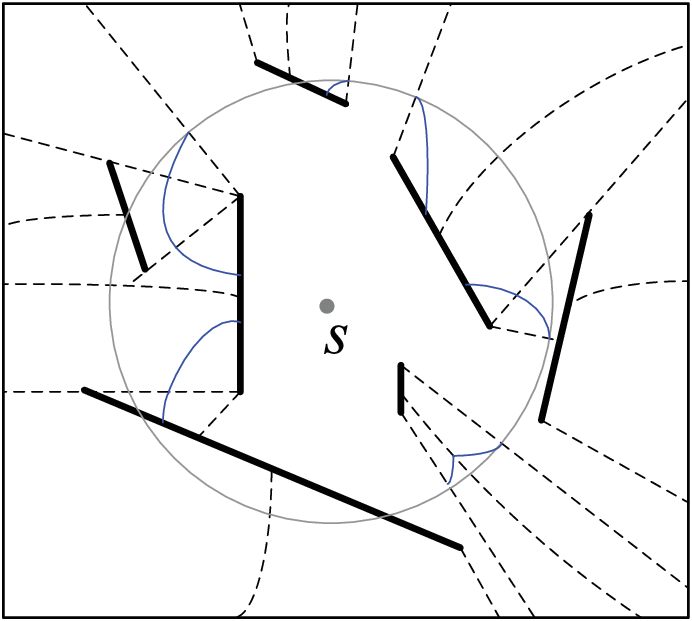}}
  \subfigure[\scriptsize { }]{\label{fig:6n:d}
    \includegraphics[scale=.36]{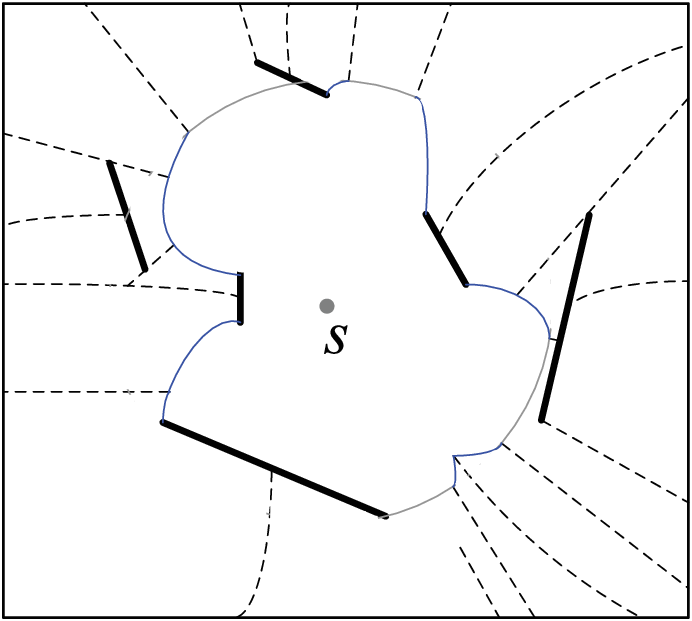}} 
 \caption{\small Illustration of   $SPM(s)$. (a) The kernel-region. (b) The circular kernel-region $\mathscr R_{map}$ actually equals the circular visibility region $\mathscr R_{cvr}(s,l)$. (c) No duplicate region is needed to be handled. (d) The merged result.} 
 \label{fig:6n}
\end{figure*}

\subsection{Regions of $SPM(s)$}\label{subsec: regions of spm}
As mentioned in Section \ref{subsec:overview of short path map}, for  two different points $p$ and $p^\prime$ in the same region, their short paths (i.e., $\pi(s,p)$ and $\pi(s,p^\prime)$) pass through the same sequence of obstacle vertices. We say the final obstacle vertex (among the sequence of obstacle vertexes) is the \textit{control point} of this region.    Let $e_c$ be a control point, we say $e_c$ is a \textit{valid control point} if $\pi(s,e_c)<l$; otherwise, we say it is an \textit{invalid control point}. 

For ease of discussion, we say the region containing the starting (source) point $s$ is the \textit{kernel-region}, and use $\mathscr R_{map}(k)$ to denote this region. We say any other region is the \textit{ordinary region}, and use $\mathscr R_{map}(o)$ to denote an ordinary region of $SPM(s)$. 
Intuitively,  the kernel-region  can have $\Omega(n)$ edges in the worst case (see e.g., Figure \ref{fig:6n:a}),  and  the number of edges of an ordinary region  has the constant descriptive complexity.

Moreover, we say the intersection set of the kernel-region $\mathscr R_{map}(k)$ and the circle $C(s,l)$ is the \textit{circular kernel-region}, and denote it as  $\mathscr R_{map}^*(k)$. Assume that $c_e$ is a valid control point of an ordinary region $\mathscr R_{map}(o)$, we say  the intersection set of this ordinary region and the circle $C(s,l-\pi(s,c_e))$ is the \textit{circular ordinary region}, and denote it as $\mathscr R_{map}^*(o)$. According to the definition of $SPM(s)$, $R_{map}^*(k)$ and $R_{map}^*(o)$, we can easily build the following theorem.

\begin{theorem}\label{theorem:a new reduction}
Given $SPM(s)$ and the maximum path length $l$, without loss of generality, assume that there are a set $\Psi$ of  ordinary regions (among all the ordinary regions) such that each of these regions has the valid control point $e_c$ (i.e., $l-\pi(s,e_c)>0$), implying that there are  a number $|\Psi|$ of \textit{circular ordinary regions}.  Let $\Psi ^*$ be the set of circular ordinary regions.   Then,  the achievable region $\mathscr R$ can be computed as $\mathscr R=\mathscr R_{map}^*(k)\bigcup_{\mathscr R_{map}*(o)\in  \Psi^*} \mathscr R_{map}^*(o)$.
\end{theorem}

\begin{lemma}\label{lemma:region and control point}
Given the kernel-region $\mathscr R_{map}(k)$ and the circle $C(s,l)$, computing their intersection set can be done in $O(n)$ time. 
\end{lemma}
\begin{proof}
This stems directly from the fact that the kernel-region $\mathscr R_{map}(k)$ can have $\Omega(n)$ edges in the worst case, and computing their intersections    takes linear time.
\end{proof}

\begin{lemma}\label{lemma:compute circle and ordinary region}
Given an ordinary region $\mathscr R_{map}(o)$ and its control point $e_c$, we assume that $e_c$ is a valid control point.   Then, computing the intersection set of this ordinary region and the circle $C(c_e,l-\pi(s,e_c))$ can be done in constant time. 
\end{lemma}
\begin{proof}
The proof is the similar as the one for Lemma \ref{lemma:region and control point}.
\end{proof}

\begin{lemma}\label{lemma:find control point}
Given $SPM(s)$, finding the control point of any ordinary region can be finished in $O(\log n)$ time. 
\end{lemma}
\begin{proof}
We just need to randomly choose a point in the region and execute a  \textit{Euclidean shortest path query}, the final obstacle vertex in the path can be obtained easily. The Euclidean shortest path query  can be done in $O(\log n)$ time, see Lemma \ref{lemma:point location query}. 
\end{proof}

\noindent \textbf{Discussion.}
Recall Section \ref{sec:reduction}, we mentioned two types of circular visibility regions. We remark that the circular kernel-region $\mathscr R_{map}^*(k)$ actually equals the circular visibility region $\mathscr R_{cvr}(s,l)$, see Figure \ref{fig:6n:b}. However, the circular ordinary region $\mathscr R_{map}^*(o)$ does not equal another type of circular visibility region $\mathscr R_{cvr}(e,l-\pi(s,e))$.  More specifically, (\romannumeral 1) $|\mathscr E  ^\prime|$ (see Theorem \ref{theorem:reduce}) is usually less than $|\Psi^*|$ (see Theorem \ref{theorem:a new reduction}); and (\romannumeral 2) the intersection set of two circular visibility regions may be non-empty (i.e., they may have  the \textit{duplicate region}), but the intersection set of any two circular ordinary regions (or, any circular ordinary region and circular kernel-region) is empty (see e.g., Figure \ref{fig:6n:c}), implying that no duplicate region is needed  to be removed, hence in theory, we  can  directly output all the circular ordinary regions and the circular kernel-region. The output shall be a set of conic polygons, since the boundaries of some circular ordinary regions possibly consist of  not only circular arc  and straight line segments but also hyperbolas. But we should note that  Section \ref{sec:problem definition} previously has stated a constraint --- the output of the algorithm to be developed is the well-organized boundaries of the achievable region $\mathscr R$ (just like shown in Figure \ref{fig:6n:d}), rather than a set of out-of-order segments, implying that we  need to handle   edges (or segments) of those conic polygons. Even so, we still can obtain an $O(n\log n)$ worst case upper bound, since the number of segments among all these conic polygons has only  complexity $O(n)$.   

\subsection{The algorithm}
The final algorithm is shown in Algorithm 3.  Its correctness directly follows from  Corollary \ref{corollary:short path map}, Theorem \ref{theorem:spm to get circle}, and Theorem \ref{theorem:a new reduction}.

\begin{algorithm}[h]
\caption{ {Finding Achievable Region of $\mathscr M$}} 
\label{alg:overall algorithm 3} 
\begin{algorithmic}[1] 
\REQUIRE  $\mathscr O$, $s$, $l$
\ENSURE $\mathscr R$
    { \small
        \STATE  Set $\mathscr R=\emptyset$
        \STATE  Construct $SPM(s)$
        \STATE  Obtain  $\mathscr R_{map}^*(k)$
        \FOR {each ordinary region $\mathscr R_{map}(o)$}
            \STATE Obtain the \textit{control point} of  $\mathscr R_{map}(o)$
            \IF{it is  a \textit{valid control point} } 
                \STATE Obtain  $\mathscr R_{map}^*(o)$
            \ENDIF
        \ENDFOR 
        \STATE Let $\mathscr R=\mathscr R_{map}^*(k)\bigcup_{\mathscr R_{map}*(o)\in  \Psi^*} \mathscr R_{map}^*(o)$ // i.e., merge all the regions obtained before
        \RETURN $\mathscr R$
        }
    \end{algorithmic}
\end{algorithm}

\begin{theorem}\label{theorem:algorithm 3 running time}
The running time of Algorithm \ref{alg:overall algorithm 3} is $O(n\log n)$.
\end{theorem}
\begin{proof}
We can easily see that, Lines 2 and  3 take $O(n\log n)$ and linear time respectively, see Corollary \ref{corollary:short path map} and Lemma \ref{lemma:region and control point}. In the \textbf{for} circulation, Lines 5 and 6 take $O(\log n)$ time, see Lemma \ref{lemma:find control point}.   We remark that Line 6   actually is (almost) the  same as the operation  --- determining if $C(e,l-(s,e))$ is a valid circle, see Theorem \ref{theorem:spm to get circle}. Moreover, Line 7 takes constant time, see Lemma \ref{lemma:compute circle and ordinary region}. Note that, the number of ordinary regions is the linear-size complexity, which stems directly from Lemme \ref{lemma:spm has linear complexity}. So, the overall execution time of the \textbf{for} circulation is also $O(n\log n)$. Finally, Line 8 is a simple  \textit{boolean set operation}, which can be done  in  $O(n\log n)$ time, since the number of edges of all these conic polygons has complexity $O(n)$,  arranging all these segments and  appealing to the algorithm in \cite{EricBerberich:aComputational} can immediately produce their union in $O((n+i)\log n)$ time, where $i$ is the number of intersections among all these segments. Note that in the context of our concern, $i$ has the linear complexity, see e.g., Figure \ref{fig:6n:c} for an illustration.  This completes the proof.  
\end{proof}


\noindent \textbf{Summary.} This section presented our final algorithm,  which significantly improves the previous ones. We remark that, maybe there still exist  more efficient solutions to improve some sub-steps in Algorithm 3, but it is obvious that  beating this worst case upper bound is (almost) impossible, this is mainly because, the nature of our problem decides  any solution to be developed, or the ones proposed in this paper, (almost) cannot be free from computing the geodesic distance, i.e., the shortest path length in the presence of obstacles.    Moreover, we remark that although  our attention is focused on the case of disjoint line-segment obstacles in this paper, it is not difficult to  see that all these algorithms can be easily applied to the case of disjoint polygonal obstacles (directly or after the minor modifications). Finally,  the $O(n\log n)$ algorithm is to construct  $SPM(s)$ with respect to all the obstacles, sometimes the maximum path length $l$ is possibly pretty small,  an output-sensitive algorithm can be easily developed by a straightforward extension of this $O(n\log n)$ algorithm.

\section{Concluding remarks} \label{sec:conclusion}
This paper proposed and studied  the FAR problem. In particular, we focused our attention to the case of line-segment obstacles. We first presented a simpler-version algorithm for the sake of intuition, which runs in  $O(n^3)$ time. The basic idea of this algorithm is to reduce our problem to computing the union of a series of circular visibility regions (CVRs). We  demonstrated its correctness, analysed its dominant steps, and  improved it by  appealing to the shortest path map (SPM) technique, which was previously used to compute the Euclidean shortest path among polygonal obstacles. We showed Hershberger-Suri's  method can be equivalently used to compute the SPM in  the case of our concern, and thus immediately yielded an $O(n^2\log n)$ algorithm. Owing to the  realization above, the third algorithm also used this technique. It however, did not construct the CVRs. Instead, it directly traversed each region of the SPM to trace the boundaries, thus obtained the  $O(n\log n)$ worst case upper bound.

We conclude this paper with several open problems.
\begin{enumerate*}
\item The dynamic version of this problem is that, if the maximum path length $l$ is not constant, how to efficiently \textit{maintain}  the dynamic achievable region?
\item The inverse problem  is that, given a closed  region ${\mathscr R}^*$, how  to \textit{efficiently} determine whether or not  ${\mathscr R}^*$ is the \textit{real} achievable region ${\mathscr R}$?
\item The multi-object version of this problem is that, if there are multiple moving objects,  how to \textit{efficiently} find their common part of their achievable regions? 
\end{enumerate*}


\appendix

 \begin{figure*}[t]
   \centering
   \subfigure[\scriptsize {  } ]{\label{fig:1N:a}
      \includegraphics[scale=.55]{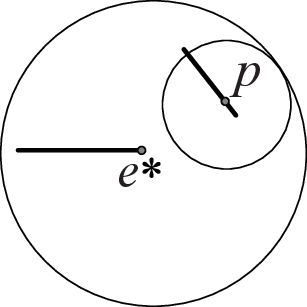}} 
      \hspace{4ex}
   \subfigure[\scriptsize { }]{\label{fig:1N:b}
      \includegraphics[scale=.6]{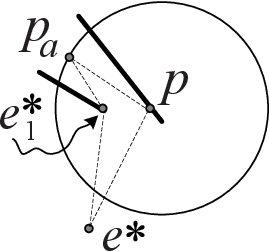}}
      \hspace{2ex}
   \subfigure[\scriptsize { }]{\label{fig:1N:c}
      \includegraphics[scale=.6]{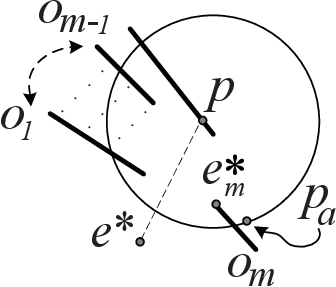} }
      \hspace{2ex}
   \subfigure[\scriptsize { }]{\label{fig:1N:d}
      \includegraphics[scale=.6]{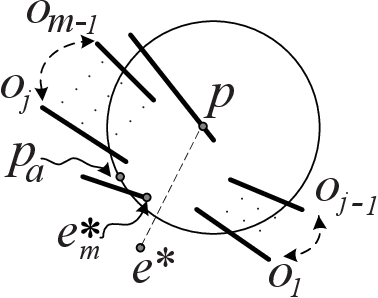} }
  \caption{\small Illustration of Lemma \ref{lemma:reduce to cvr}. (a) The small and big circles denote ${ C}(p,l-\pi(s,p))$ and  ${C}(e^*,l-\pi(s,e^*))$, respectively. In  (b-d), we omit the circle ${C}(e^*,l-\pi(s,e^*))$ for clearness. } 
  \label{fig:1N}
 \end{figure*}

\noindent\textbf{A. The proof of Lemma \ref{lemma:reduce to cvr}} 

\begin{proof} 
There are several cases for a point $p$ $\in C(s,l)$ such that $\pi(s,p)<l$.

{$\bullet$} Case 1:   $\underline{p\in \mathscr E^\prime}$. In this case, it is obvious that  ${\mathscr R}_{cvr}(p, l-\pi(s,p))\subset \bigcup_{e\in \mathscr E^\prime} {\mathscr R}_{cvr}(e, l-\pi(s,e))\bigcup {\mathscr R}_{cvr}(s,l)$.  

{$\bullet$} Case 2: $p$ is located in \underline{ one of obstacles but $p\notin \mathscr E^\prime$}  .   
Without loss of generality,  assume that $e^*$ is the previous point on the shortest path from $s$ to $p$ such that $e^*\in \{s\bigcup \mathscr E^\prime\}$.   Consider the two circles $C(p,l-\pi(s,p))$ and $C(e^*,l-\pi(s,e^*))$. It is easy to know that $dist(p,e^*)+ (l-\pi(s,p))=l-\pi(s,e^*)$  and $p\in C(e^*, l-\pi(s,e^*))$. Hence,  ${ C}(p,l-\pi(s,p))$ must be an inscribed circle of  ${ C}(e^*,l-\pi(s,e^*))$. This implies that ${ C}(p, l-\pi(s,p))\subset { C}(e^*,l-\pi(s,e^*))$. Therefore, we get \textbf{a preliminary conclusion} --- for any point $p_a$ such that $p_a\in {\mathscr R}_{cvr}(p,l-\pi(s,p))$ and $\sphericalangle (p_a, e^*)$, we have that $p_a\in {\mathscr R}_{cvr}(e^*,l-\pi(s,e^*))$.

{\large $\star$} Case 2.1: There is  \underline{\textit{no other obstacle}}  that makes impact on the size of ${\mathscr R}_{cvr}(p, l-\pi(s,p))$. See  Figure \ref{fig:1N:a}.  Clearly, for any point  $p_a\in \mathscr R_{cvr}(p,l-\pi (s,p))$, we have  $\sphericalangle (p_a,e^*)$.  By  the preliminary conclusion shown in the previous paragraph, this completes {the proof of Case 2.1}. 

{\large $\star$} Case 2.2: 
There are  \underline{\textit{{other} obstacles}}  that make impact on the size of ${\mathscr R}_{cvr}(p, l-\pi (s,p))$.  
The key of point is to prove that,  for any point $p_a$ such that  $p_a\in {\mathscr R}_{cvr}(p,  l-\pi (s,p) )$ and $\lnot(\sphericalangle (p_a,e^*))$, it must be located in a circular visibility region whose center is the endpoint of certain obstacle.
Let $\mathscr O^\prime$ be the set of other obstacles that make impact on the size of ${\mathscr R}_{cvr}(p,  l-\pi (s,p) )$.  
For ease of discussion,  assume that $e^*_i$ and $e^\prime _i$ are the  endpoints of the $i$th obstacle (among $|\mathscr O^\prime|$ obstacles).  Let  $e^*_i$  denote the  endpoint  such that $\pi(e^*,e^*_i)\leq\pi(e^*,e^\prime_i)$. Let $m^\prime=|\mathscr O^\prime|$, and  $m$ be an arbitrary integer. We next prove \textit{by induction} that \textbf{the following proposition} called $\mathbb{P}$ holds --- for any point $p_a$ such that $p_a\in {\mathscr R}_{cvr}(p,l-\pi(s,p))$ and    $\lnot (\sphericalangle (p_a,e^*))$, we have that $p_a\in \bigcup _{i=1} ^{m^\prime} {\mathscr R}_{cvr}(e^*_i,  l-\pi(s,e^*_i) )$.

We first consider \underline{$m^\prime=1$}.  We connect the following points, $e^*$, $p$, $p_a$ and $e^*_1$. Then, they build a circuit  with four edges (see Figure \ref{fig:1N:b}). Let $\Delta$ be ($dist(e^*,p)$+$dist(p,p_a)$)$-$($dist(e^*,e^*_1)$+ $dist(e^*_1,p_a)$).  According to \textit{analytic geometry} and \textit{graph theory}, it is easy to know that $\Delta>0$.    This implies that the radius of ${\mathscr R}_{cvr}(e^*_1, l-\pi(s,e^*_1)  )$ is equal to $dist(e_1^*,p_a)+ \Delta$. So,  for any point $p_a$ such that $p_a\in {\mathscr R}_{cvr}(p, l-\pi(s,p))$ and $\lnot (\sphericalangle (p_a,e^*))$, we have that   $p_a\in {\mathscr R}_{cvr}(e^*_1, l-\pi(s,e^*_1) )$. Therefore, the proposition $\mathbb{P}$  holds when $m^\prime=1$.

By convention, we assume   $\mathbb{P}$ holds when \underline{$m^\prime=m-1$}. We next show it  also holds when \underline{$m^\prime=m$}.  Let $o_1$,  $\cdots$, $o_m$ denote  these obstacles, i.e., $\mathscr O^\prime=\{o_1,\cdots,o_m\}$. We remark that (\romannumeral 1) it corresponds to ``$m^\prime=m-1$''  if $o_m$  viewed from  $p$   is totally \textit{blocked}  by other $m-1$ obstacles; and (\romannumeral 2) in the rest of the proof, unless stated otherwise, we use ``viewed from  $p$'' by default when the location relation of obstacles is considered. There are three  cases. 

First, if $o_m$ is   \textit{disjointed} with other $m-1$ obstacles, we denote by \underline{$\asymp (o_m,\bigcup _{i=1}^{m-1} o_i)$} this case. See  Figure \ref{fig:1N:c}.  Let's consider $o_m$,  according to the method for proving the case $m^\prime=1$,  it is easy to get \textbf{a  result} --- for any point $p_a$ such that $p_a\in {\mathscr R}_{cvr}(p, l-\pi(s,p) )$ and    $p_a\notin \bigcup _{i=1} ^{m-1} {\mathscr R}_{cvr}(e^*_i, l-\pi(s,e^*_i))$ and $\lnot (\sphericalangle (p_a,e^*))$, we have that  $p_a\in   {\mathscr R}_{cvr}(e^*_m, l-\pi (s,e^*_m) )$. Furthermore, we have assumed $\mathbb{P}$ holds when $m^{\prime}=m-1$. This completes the proof of the case  $\asymp (o_m,\bigcup _{i=1}^{m-1} o_i)$.


Second, if $o_m$ is in the front of other $m-1$ obstacles, we denote by \underline{$\succ (o_m,\bigcup _{i=1}^{m-1} o_i)$} this case.
We connect the points $e^*$, $\cdots$, $e_i^*$, $\cdots$, and $e_m^*$ such that the set of segments build the shortest path from $e^*$ to $e_m^*$. The total length of these segments is $\pi(e^*,e_m^*)$.  Without loss of generality, assume that $p_a$ is to be a point such that $p_a\in {\mathscr R}_{cvr}(p,l-\pi(s,p) $ and $p_a\notin \bigcup _{i=1} ^{m-1} {\mathscr R}_{cvr}(e^*_i, l-\pi(s,e^*_i) )$ and $\lnot (\sphericalangle (p_a,e^*))$. We also  connect the points $e^*_m$ and $p_a$. Naturally, we get the shortest path from $e^*$ to $p_a$, its total length is $\pi(e^*,e_m^*)+ dist(e_m^*,p_a)$. This implies that there is no other path  (from $e^*$ to $p_a$) whose length is less  than $\pi(e^*,e_m^*)+ dist(e_m^*,p_a)$.   Let $\Delta$ be  $( \pi(e^*,p)+dist(p,p_a))-(\pi(e^*,e_m^*)+ dist(e_m^*,p_a))$, we have that $l-\pi(s,e_m^*)= dist(e_m^*,p_a) +\Delta > dist(e_m^*,p_a)$, since $\Delta>0$. Therefore, $p_a\in   {\mathscr R}_{cvr}(e^*_m, l-\pi(s,e^*_m))$. Furthermore, we have assumed $\mathbb{P}$ holds when $m^{\prime}=m-1$. 
This completes the proof of the case $\succ (o_m,\bigcup _{i=1}^{m-1} o_i)$.

Third, if $o_m$ is partially blocked by other $m-1$ obstacles. We denote by \underline{$\prec(o_m,\bigcup _{i=1}^{m-1} o_i)$} this case. Without loss of generality, assume that $o_m$  is partially blocked by an obstacle $o_j$. 
\underline{(\romannumeral 1)} If  $o_j$ does not block any other $m-2$ obstacles (see, e.g., Figure \ref{fig:1N:d}), according to the method for proving  the case $m^\prime=1$,  we can also get a result which is the same as the  result  shown in the case $\asymp (o_m,\bigcup _{i=1}^{m-1} o_i)$. 
\underline{(\romannumeral 2)} Otherwise, we connect the points $e^*$, $\cdots$, $e_i^*$ ($i\in [1,\cdots,j-1]$), $\cdots$, and $e_m^*$ such that the set of segments build the shortest path from $e^*$ to $e_m^*$. The rest of steps are the same as the ones for proving the case $\succ(o_m,\bigcup _{i=1}^{m-1} o_i)$. And we  can also get a result which is the same as the previous one. Furthermore, we have assumed the proposition $\mathbb{P}$ holds when $m^{\prime}=m-1$. This completes the proof of the case $\prec(o_m,\bigcup _{i=1}^{m-1} o_i)$. 

\underline{In summary},  the proposition $\mathbb{P}$ also holds when $m^\prime=m$. Combining the preliminary conclusion shown in the first  paragraph of  Case 2, this completes {the proof of Case 2.2}.

{$\bullet$} Case 3: $p$ is not located  in any obstacle. The proof  for this case is almost the same as the one for  Case 2. (Substituting the words ``other obstacles'' in  Case 2 with ``obstacles''.)  
Pulling all together, hence the lemma holds.
\end{proof}


{ \normalsize
\bibliographystyle{abbrv}
\bibliography{sample}
}


\end{document}